\documentclass[a4paper]{amsart}
\usepackage[utf8]{inputenc}
\usepackage{amssymb}

\theoremstyle{plain}

  \newtheorem{prop}{Proposition}
  \newtheorem*{cor}{Corollary}
  
\theoremstyle{definition}

  \newtheorem{rem}{Remark}

\newcommand{\mf}{\mathfrak}
\newcommand{\mc}{\mathcal}
\newcommand{\on}{\operatorname}
\newcommand{\g}{\mathfrak{g}}
\newcommand{\h}{\mathfrak{h}}
\newcommand{\dd}{\mathfrak{d}}

\newcommand{\la}{\langle}
\newcommand{\ra}{\rangle}
\newcommand{\R}{\mathbb{R}}
\newcommand{\C}{\mathbb{C}}

\newcommand{\half}{\frac{1}{2}}

\title{On integrability of 2-dimensional $\sigma$-models of Poisson-Lie type}
\author{Pavol Ševera}
\address{Section of Mathematics, University of Geneva, Switzerland}
\email{pavol.severa@gmail.com}
\thanks{Supported in part by  the grant MODFLAT of the European Research Council and the NCCR SwissMAP of the Swiss National Science Foundation.}

\begin{document}
\maketitle
\begin{abstract}
We describe a simple procedure for constructing a Lax pair for suitable 2-dimensional $\sigma$-models appearing in Poisson-Lie T-duality
\end{abstract}
\section{Introduction}
There is a class of 2-dimensional $\sigma$-models, introduced in the context of Poisson-Lie T-duality \cite{ks}, whose solutions are naturally described in terms of certain flat connections. The target space of such a $\sigma$-model is $D/H$, where $D$ is a Lie group and $H\subset D$ a subgroup. The $\sigma$-model is defined by the following data: an invariant symmetric non-degenerate pairing $\la,\ra$ on the Lie algebra $\dd$ such that the Lie subalgebra $\h\subset\dd$ is Lagrangian, i.e. $\h^\perp=\h$, and  a subspace $V_+\subset\dd$ such that $\dim V_+=(\dim\dd)/2$ and such that $\la,\ra|_{V_+}$ is positive definite. The construction and properties of these $\sigma$-models are recalled in Section \ref{sec:PLTD} (including the Poisson-Lie T-duality, which says that the $\sigma$-model, seen as a Hamiltonian system, is essentially independent of $H$). Let us call them \emph{$\sigma$-models of Poisson-Lie type}.

The solutions $\Sigma\to D/H$ of equations of motion of such a $\sigma$-model can be encoded in terms of $\dd$-valued 1-forms $A\in\Omega^1(\Sigma,\dd)$ satisfying 
\begin{subequations}\label{sols}
\begin{equation}\label{Aflat}
dA+[A,A]/2=0
\end{equation}
\begin{equation}
A\in\Omega^{1,0}(\Sigma, V_+) \oplus \Omega^{0,1}(\Sigma,V_-),
\end{equation}
\end{subequations}
where $V_-:=(V_+)^\perp\subset\dd$.
Namely, the flatness \eqref{Aflat} of $A$ implies that there is a map $\ell:\tilde\Sigma\to D$ (where $\tilde\Sigma$ is the universal cover of $\Sigma$) such that $A=-d\ell\,\ell^{-1}$. If the holonomy of $A$ is in $H$ then $\ell$ gives us a well-defined map $\Sigma\to D/H$. The maps $\Sigma\to D/H$ obtained in this way are exactly the solutions of equations of motion.

%One can observe that $H$ enters in a limited way in this picture - only in the holonomy constraint. Indeed, if we impose that the holonomy is 1, i.e.\ perform a certain reduction of the phase space of the $\sigma$-model, the reduced phase space is independent of $H$. This statement is Poisson-Lie T-duality (in the case of no spectators) - if $\h,\h'\subset\dd$ are two Lagrangian Lie subalgebras then the resulting $\sigma$-models with the targets $D/H$ and $D/H'$ are isomorphic as Hamiltonian systems (after we perform the reductions).

As first observed by Klimčík \cite{kl}, and later by Sfetsos \cite{sf}, and Delduc, Magro, and Vicedo \cite{dmv2}, some $\sigma$-models of Poisson-Lie type are integrable. Their integrability is proven by finding a Lax pair, i.e.\ a 1-parameter family of flat connections (with parameter $\lambda$)
$$A_\lambda\in\Omega^1(\Sigma,\g)\qquad dA_\lambda+[A_\lambda,A_\lambda]/2=0$$
where $\g$ is a suitable semisimple Lie algebra. Such a family is constructed for every element of the phase space, i.e.\ for every $A\in\Omega^1(\Sigma,\dd)$ satisfying \eqref{sols}.

The aim of this note is to make the construction of $A_\lambda$ transparent. We simply observe that if $A\in\Omega^1(\Sigma,\dd)$ satisfies \eqref{sols} and if $p:\dd\to\g$ is a linear map such that
$$[p(X),p(Y)]=p([X,Y])\quad\forall X\in V_+,\;Y\in V_-$$  then 
$$d\,p(A)+[p(A),p(A)]/2=0.$$
A suitable family $p_\lambda:\dd\to\g$ will then give us a family of flat connections
$$A_\lambda=p_\lambda(A).$$

As an example, we provide a very simple construction of such families $p_\lambda$ in the case when $\dd=\g\otimes W$, where $W$ is a 2-dimensional commutative algebra. These families recover the deformations of the principal chiral model from \cite{dmv2,kl,sf}. Our purpose is thus modest - it is simply to clarify previously constructed integrable $\sigma$-models. There is possibly a less naive construction of families $p_\lambda$ that might produce new integrable models, but we leave this question open. 

\section{$\sigma$-models of Poisson-Lie type and Poisson-Lie T-duality}\label{sec:PLTD}
In this section we review the properties of the ``2-dimensional $\sigma$-models of Poisson-Lie type'' introduced in \cite{ks} (together with their Hamiltonian picture from \cite{ks-ham} and using the target spaces of the form $D/H$, as introduced in \cite{ks-D/H}).

Let $\dd$ be a Lie algebra with an invariant non-degenerate symmetric bilinear form $\la,\ra$ of symmetric signature  and let $V_+\subset\dd$ be a linear subspace with $\dim V_+=(\dim\dd)/2$, such that $\la,\ra|_{V_+}$ is positive-definite. 

Let $M=D/H$ where $D$ is a connected Lie group integrating $\dd$ and $H\subset D$ is a closed connected subgroup such that its Lie algebra $\h\subset\dd$ is Lagrangian in $\dd$.

This data defines a Riemannian metric $g$ and a closed 3-form $\eta$ on $M$. They are given by
\begin{align*}
g(\rho(X),\rho(Y))&=\half\la X,Y\ra\qquad\forall X,Y\in V_+\\
p^*\eta&=-\half\eta_D+\half d\la\mc A,\theta_L\ra
%\eta(\rho(X),\rho(Y),\rho(Z))&=-\half\la[X,Y],Z\ra+\half\bigl(\la P[X,Y],Z\ra + c.p.\bigr)
\end{align*}
Here $\rho$ is the action of $\dd$ on $M=D/H$, $p:D\to D/H$ is the projection, $\eta_D\in\Omega^3(D)$ is the Cartan 3-form (given by $\eta_D(X^L,Y^L,Z^L)=\la[X,Y],Z\ra$ ($\forall X,Y,Z\in\dd$)), $\theta_L\in\Omega^1(D,\dd)$ is the left-invariant Maurer-Cartan form on $D$ (i.e.\ $\theta_L(X^L)=X$), and $\mc A\in\Omega^1(D,\h)$ is the connection on the principal $H$-bundle $p:D\to D/H$ whose horizontal spaces are the right-translates of $V_+$.\footnote{The conceptual definition of $g$ and $\eta$ is as follows: the trivial vector bundle $\dd\times M\to M$ is naturally an exact Courant algebroid, with the anchor given by $\rho$ and the Courant bracket of its constant sections being the Lie bracket on $\dd$. Then $V_+\times M\subset\dd\times M$ is a generalized metric, which  is equivalent to the metric $g$ and the closed 3-form $\eta$. We shall not use this language in this paper, in order to keep it short.}
%
%Equivalently (in the case of $M=D/H$), the 3-form $\eta$ can be computed as follows. If we right-translate $V_+\subset\dd$ over $D$ and interpret the resulting subbundle of $TD$ as the horizontal distribution of a connection $\mathcal A\in\Omega^1(D,\h)$ on the principal $H$-bundle $p:D\to D/H$, then
%$$p^*\eta=\eta_D+\half\la\theta_L,\mathcal A\ra$$
%where $\eta_D$ is the bi-invariant closed 3-form on $D$ given by $\eta_D(X,Y,Z)=-\half\la[X,Y],Z\ra$ ($\forall X,Y,Z\in\dd$) and $\theta_L\in\Omega^1(D,\dd)$ is the left-invariant Maurer-Cartan form. From this description it is clear that $\eta$ is closed.

The metric $g$ and the 3-form $\eta$ then define a $\sigma$-model with the standard action
 functional
$$S(f)=\int_\Sigma g(\partial_+f,\partial_-f)+\int_Yf^*\eta$$
where $\Sigma$ is (say) the cylinder with the usual metric $d\sigma^2-d\tau^2$ and $f:\Sigma\to M$ is a map extended to the solid cylinder $Y$ with boundary $\Sigma$.

For our purposes, the main properties of these $\sigma$-models are the following:

\begin{itemize}
\item The solutions of the equations of motion are in (almost) 1-1 correspondence with 1-forms $A\in\Omega^1(\Sigma,\dd)$ satisfying \eqref{sols}. More precisely, a map $f:\Sigma\to M$ is a solution iff it admits a lift $\ell:\tilde\Sigma\to D$ such that $A:=-d\ell\,\ell^{-1}$ satisfies \eqref{sols}. Notice that $A$ is uniquely specified by $f$ (the lift $\ell$ is not unique - it can be multiplied by an element of $H$ on the right).
\item When we restrict $A$ to $S^1\subset\Sigma=S^1\times\R$, we get a 1-form $j(\sigma)d\sigma\in\Omega^1(S^1,\dd)$. The $\dd$-valued functions $j(\sigma)$ on the phase space of the sigma model satisfy the current algebra Poisson bracket
\begin{equation}\label{curr}
\{j_a(\sigma),j_b(\sigma')\}=f_{ab}^c\, j_c(\sigma)\,\delta(\sigma-\sigma')+t_{ab}\,\delta'(\sigma-\sigma')
\end{equation}
(written using a basis $e^a$ of $\dd$, with $f_{ab}^c$ being the structure constants of $\dd$ and $t_{ab}$ the inverse of the matrix of $\la e^a,e^b\ra$).
The Hamiltonian of the $\sigma$-model is
\begin{equation}\label{ham}
\mathcal H=\half\int_{S^1}\la j(\sigma),Rj(\sigma)\ra\,d\sigma
\end{equation}
where $R:\dd\to\dd$ is the reflection w.r.t.\ $V_+$.
%\item The phase space of the $\sigma$-model (i.e.\ $T^*(LM)$ with the symplectic form modified by $\eta$)   can be identified  with a certain moduli space of flat connections, namely with the space of flat $\dd$-connection on an annulus that restrict to an $\h$-connection on the inner boundary circle, modulo the group of gauge transformation that vanish on the outer circle and take values in $H$ on the inner circle. In this description $j(\sigma)d\sigma$ is the restriction of the flat connection on the annulus to the outer circle. 
%
%\item We can reduce the phase space by imposing the holonomy along the inner circle to be 1. This reduced Hamiltonian system is clearly independent of $H$. This statement is Poisson-Lie T-duality - it is an isomorphism of (suitable reductions) of $\sigma$-models with the same $\dd$ and $V_+$, but with different $H$'s.
\end{itemize}

Finally, let us observe that the phase space of the $\sigma$-model depends on the choice of $H\subset D$ only mildly; when we impose the constraint that $A$ has unit holonomy, the reduced Hamiltonian system is independent of $H$. This statement is the Poisson-Lie T-duality (in the case of no spectators). (In more detail, the phase space of the $\sigma$-model is the space of maps $\ell:\R\to D$ which are quasi-periodic in the sense that for some $h\in H$ we have $\ell(\sigma+2\pi)=\ell(\sigma)h$, modulo the  action of $H$ by right multiplication. The reduced phase space is $(LD)/D$ (i.e.\ periodic maps modulo the action of $D$); it is the subspace of $\Omega^1(S^1,\dd)$ given by the unit holonomy constraint.)

\section{Constructing new flat connections}

As we have seen, the solutions of our $\sigma$-model give rise to flat connections $A\in\Omega^1(\Sigma,\dd)$ satisfying \eqref{sols}.
We can obtain new flat connections out of $A$ using the following simple observation, which is also the main idea of this paper.

\begin{prop}\label{prop:p}
Let $\g$ be a Lie algebra and let $p:\dd\to\g$ be a linear map such that
\begin{equation}\label{V-morph}
[p(X),p(Y)]=p([X,Y])\quad\forall X\in V_+,\;Y\in V_-.
\end{equation}
If $A\in\Omega^1(\Sigma,\dd)$ satisfies \eqref{sols} then $p(A)\in\Omega^1(\Sigma,\g)$ is flat, i.e.
$$d\,p(A)+[p(A),p(A)]/2=0.$$
\end{prop}
\begin{proof}
Let us use the following notation: for $\alpha\in\Omega^1(\Sigma)$ let $\alpha^+\in\Omega^{1,0}(\Sigma)$ and $\alpha^-\in\Omega^{0,1}(\Sigma)$ denote the components of $\alpha$, i.e.\ $\alpha=\alpha^++\alpha^-$. In particular, $A^+\in\Omega^{1,0}(\Sigma, V_+)$ and $A^-\in\Omega^{0,1}(\Sigma, V_-)$. We then have
\begin{multline*}
d\,p(A)+[p(A),p(A)]/2=d\,p(A)+[p(A^+),p(A^-)]=p(dA+[A^+,A^-])\\=p(dA+[A,A]/2)=0.
\end{multline*}
\end{proof}

Given a 1-parameter family of maps $p_\lambda:\dd\to\g$  satisfying \eqref{V-morph} we would thus get a 1-parameter family of flat connections $A_\lambda=p_\lambda(A)$ on $\Sigma$, which may then be used to show integrability of the model. Let us observe that the Poisson brackets of the ``Lax operators'' $L(\sigma,\lambda):=p_\lambda(j(\sigma))$ are automatically of the form considered in \cite{m1,m2} (i.e.\ containing a $\delta(\sigma-\sigma')$ and a $\delta'(\sigma-\sigma')$ term), and so one can in principle extract an infinite family of Poisson-commuting integrals of motion out of the holonomy of $A_\lambda$.

\begin{rem}
The procedure of finding integrable  deformations of integrable $\sigma$-models, due to Delduc, Magro, and Vicedo \cite{dmv}, can be rephrased in our formalism as follows. Suppose that for some particular pair $V_+\subset\dd$ we find a family $p_\lambda:\dd\to\g$ showing integrability of the model. Let us deform the Lie bracket on $\dd$, and possibly the pairing $\la,\ra$, in such a way that the restriction of the Lie bracket to $V_+\times V_-\to\dd$ is undeformed. Then the same family $p_\lambda$ will satisfy \eqref{V-morph} also for the deformed structure on $\dd$ and show  integrability of the deformed model. These deformations of $\dd$ do not change the system \eqref{sols} (and if $\la,\ra$ is not deformed then they don't change the Hamiltonian \eqref{ham} either), but they do change the Poisson structure \eqref{curr} on the phase space.
\end{rem}

\begin{rem}
There is a version of $\sigma$-models of Poisson-Lie type, introduced in \cite{ks-coset}, with the target space if $F\backslash D/H$, where $\mf f\subset \dd$ is an isotropic Lie algebra (and one needs to suppose that $F$ acts freely on $D/H$). In this case $V_+\subset\dd$ is required to be such that $\la,\ra|_{V_+}$ is semi-definite positive with kernel $\mf f$ (in particular, $\mf f\subset V_+$), and such that $[\mf f,V_+]\subset V_+$ (we still have $\dim V_+=(\dim \dd)/2$). The phase space is the Marsden-Weinstein reduction of $\Omega^1(S^1,\dd)$ by $LF$, i.e.\ $\Omega^1(S^1,\mf f^\perp)/LF$. The solutions of equations of motion are still given by the solutions of \eqref{sols}, though this time $A$ is defined only up to $F$-gauge transformations. In this case we can still use Proposition \ref{prop:p} without any changes. This setup should cover, in particular, the discussion of symmetric spaces in \cite{dmv}.
\end{rem}

\section{Getting a Lax pair in a simple case}

In this section we give a simple example of pairs $V_+\subset\dd$ with natural 1-parameter families $p_\lambda$ satisfying \eqref{V-morph}.

Let $\g$ be a Lie algebra with an invariant inner product $\la,\ra_\g$ and let $W$ be a 2-dimensional commutative associative algebra with unit. ($W$ is isomorphic to one of $\C$, $\R\oplus\R$, $\R[\epsilon]/(\epsilon^2)$.)
Let
$$\dd:=\g\otimes W$$
with the Lie bracket $[X_1\otimes w_1,X_2\otimes w_2]_\dd=[X_1,X_2]_\g\otimes w_1w_2.$

We choose the following additional data in $W$ to produce a pairing $\la,\ra$ on $\dd$ and a subspace $V_+\subset\dd$:

 To get the pairing, let $\theta:W\to\R$ be a linear form such that the pairing on $W$ given by $\la w_1,w_2\ra_W:=\theta(w_1w_2)$ is non-degenerate (i.e.\ such that it makes $W$ to a Frobenius algebra) and indefinite. The pairing on $\dd$ is then defined via
 $$\la X_1\otimes w_1,X_2\otimes w_2\ra:=\la X_1,X_2\ra_\g\, \theta(w_1w_2).$$
 
  To get $V_+\subset\dd$, let $V^0_+\subset W$ be a 1-dimensional subspace such that $\la,\ra_W$ is positive-definite on $V^0_+$. Let
  $$V_+=\g\otimes V^0_+.$$
   Then $V_-=\g\otimes V^0_-$ where $V^0_-=(V^0_+)^\perp$.

We can now describe the construction of a family $p_\lambda:\dd\to\g$ satisfying \eqref{V-morph}. Let us choose non-zero elements $e_+\in V^0_+$ and $e_-\in V^0_-$ (this choice is inessential).

\begin{prop}
If a linear form $q:W\to\R$ satisfies 
\begin{equation}\label{q}
q(e_+)q(e_-)=q(e_+e_-)
\end{equation}
then the map $p={\on{id}_\g}\otimes q:\dd\to\g$ satisfies \eqref{V-morph}. 
\end{prop}
\begin{proof}
If $X=x\otimes e_+\in V_+$ and $Y=y\otimes e_-\in V_-$ then $$p([X,Y]_\dd)=p\bigl([x,y]_\g\otimes e_+e_-\bigr)=q(e_+e_-)[x,y]_\g=[q(e_+)x, q(e_-)y]_\g=[p(X),p(Y)]_\g.$$
\end{proof}

The solutions $q\in W^*$ of \eqref{q} form a curve in $W^*$, which is either a hyperbola or a union of two lines. 
If 
$$e_+e_-=ae_++be_- \quad a,b\in\R,$$ 
we  rewrite \eqref{q} as
$$\bigl(q(e_+)-b\bigr)\bigl(q(e_-)-a\bigr)=ab.$$
We thus have a hyperbola if $ab\neq0$ and a union of two straight lines if $ab=0$. 

One can easily check that $ab=0$ iff one of $V^0_\pm$ is of the form $\R e$ where $e\in W$ satisfies $e^2=e$. This means that one of $V_\pm=\g\otimes V^0_\pm\subset\dd$ is a Lie subalgebra isomorphic to $\g$ and thus, according to \cite{ks-wzw}, for any Lagrangian $\h\subset\dd$, the corresponding $\sigma$-model is simply the WZW model given by $G$.

Let us now choose a rational parametrization $\lambda\mapsto q_\lambda$ of the hyperbola \eqref{q}. The standard parametrization in this context seems to be the one sending $\lambda=\pm1$ to the two points at the infinity of the hyperbola, and $\lambda=\infty$ to 0 (though any other parametrization would do). This gives
$$q_\lambda=\frac{q_+}{1+\lambda}+\frac{q_-}{1-\lambda}$$
where $q_+,q_-\in W^*$ are given by 
$$q_+(e_-)=q_-(e_+)=0,\quad q_+(e_+)=2b,\quad q_-(e_-)=2a.$$
(If the curve a union of two lines then this pametrizes only one of the lines, or possibly just a single point.)
\begin{cor}
If $A\in\Omega^1(\Sigma,\dd)$ satisfies \eqref{sols}, i.e.\ if $A=A_+\otimes e_+ + A_-\otimes e_-$ with $A_+\in\Omega^{1,0}(\Sigma,\g)$ and $A_-\in\Omega^{0,1}(\Sigma,\g)$ and if $A$ is flat, then the $\g$-connections 
$$A_\lambda=(1\otimes q_\lambda)(A)=\frac{2b}{1+\lambda}A_+ +\frac{2a}{1-\lambda}A_-$$
are flat.
\end{cor}

The $\g$-valued Lax operator obtained in this way is thus
\begin{equation}\label{lax}
L(\sigma,\lambda)=\frac{2b}{1+\lambda}j_+(\sigma) +\frac{2a}{1-\lambda}j_-(\sigma)
\end{equation}
where we decomposed $j(\sigma)$ as $j(\sigma)=j_+(\sigma)\otimes e_+ + j_-(\sigma)\otimes e_-$. For completeness, the Hamiltonian \eqref{ham} is
$$\mathcal H=\half\int_{S^1}\Bigl(\theta(e_+^2)\la j_+(\sigma),j_+(\sigma)\ra_\g-\theta(e_-^2)\la j_-(\sigma),j_-(\sigma)\ra_\g\Bigr)d\sigma.$$

\section{Examples of the example}
In this section $\g$ is a compact Lie algebra and $G$ the corresponding compact 1-connected Lie group.

Let start with the case of $W=\R\oplus\R$, i.e.\ $\dd=\g\oplus\g$. The only admissible $\theta\in W^*$, up to rescaling (which can be absorbed to $\la,\ra_\g$) and exchange of the two components of $W$, is $\theta(x,y)=x-y$. (Here the main limiting factor is existence of a lagrangian Lie subalgebra $\h\subset\dd$: if $\theta(x,y)=cx+dy$ with $cd\neq0$ (the non-degeneracy condition), it forces $c=-d$.) The pairing $\la,\ra$ on $\dd$ is $\la(X_1,X_2),(Y_1,Y_2)\ra=\la X_1,X_2\ra_\g-\la Y_1,Y_2\ra_\g$.

 We have
$$e_+=(1,t)\qquad e_-=(t,1)$$
for some $-1<t<1$. The Lax operator \eqref{lax}, written in terms of $j=(j_1,j_2)$, is
$$L(\sigma,\lambda)=\frac{2t}{(1+t)(1-t)^2}\,\Bigl(
\frac{1}{1+\lambda}\,\bigl(j_1(\sigma)-t j_2(\sigma)\bigr)
+\frac{1}{1-\lambda}\,\bigl(j_2(\sigma)-tj_1(\sigma)\bigr)
\Bigr)$$
The Poisson brackets \eqref{curr} of $j_{1,2}$ are
\begin{align*}
\{j_{1a}(\sigma),j_{1b}(\sigma')\}&=f_{ab}^c\, j_{1c}(\sigma)\,\delta(\sigma-\sigma')+\delta_{ab}\,\delta'(\sigma-\sigma')\\
\{j_{2a}(\sigma),j_{2b}(\sigma')\}&=f_{ab}^c\, j_{2c}(\sigma)\,\delta(\sigma-\sigma')-\delta_{ab}\,\delta'(\sigma-\sigma')\\
\{j_{1a}(\sigma), j_{2b}(\sigma')\}&=0
\end{align*}
and the Hamiltonian \eqref{ham} is
$$\mathcal H=\frac{1}{2(1-t^2)}\int_{S^1}\Bigl((1+t^2)\bigl(\la j_1(\sigma),j_1(\sigma)\ra_\g+\la j_2(\sigma),j_2(\sigma)\ra_\g\bigr)
-4t\la j_1(\sigma),j_2(\sigma)\ra_\g\Bigr)d\sigma.$$
The degenerate case $t=0$ (when $a=b=0$) corresponds to the WZW-model on $\g$.

The natural choice for a Lagrangian Lie subalgebra $\h\subset\dd$ is the diagonal $\g\subset\dd$. The target space of the $\sigma$-model is $D/G\cong G$. It is the so-called ``$\lambda$-deformed $\sigma$-model'' introduced by Sfetsos in \cite{sf} (Sfetsos's $\lambda$ is our $t$).

Let us now consider the case $W=\C$, which is the richest one. In this case any non-zero $\theta\in W^*$ is suitable. Let $\theta(z)=\on{Im}(e^{2i\alpha}z)$ for some $\alpha\in\R$. We thus have $\dd=\g\otimes\C=\g_\C$ (seen as a real Lie algebra) with the pairing $\la X,Y\ra=\on{Im}(e^{2i\alpha}\la X,Y\ra_{\g_\C})$ where $\la,\ra_{\g_\C}$ is the $\C$-bilinear extension of $\la,\ra_\g$.

 In this case 
$$e_+=e^{-i\alpha+i\phi}\qquad e_-=e^{-i\alpha-i\phi}$$ for some $\phi\in(0,\pi/2)$. If $e^{2i\alpha}=e^{\pm 2i\phi}$ then the resulting $\sigma$-model (regardless of the choice of $\h\subset\dd$) is the WZW-model on $G$.

The Lax operator \eqref{lax} is
$$
L(\sigma,\lambda)=\frac{1}{2\sin^2 2\phi}\Bigl(
\frac{e^{2i\alpha}-e^{-2i\phi}}{1+\lambda}+
\frac{e^{2i\alpha}-e^{2i\phi}}{1-\lambda}\Bigr)J(\sigma)
+c.c.
$$
(where $c.c.$ stands for ``complex conjugate''). Here  $J=j_\mathit{re}+ij_\mathit{im}$  and $\bar J=j_\mathit{re}-ij_\mathit{im}$ where $j_\mathit{re}$ and $j_\mathit{im}$ are given by $j=j_\mathit{re}\otimes 1+j_\mathit{im}\otimes i$,  their Poisson brackets \eqref{curr} (written in an orthonormal basis of $\g$) are
\begin{align*}
\{J_a(\sigma),J_b(\sigma')\}&=f_{ab}^c\, J_c(\sigma)\,\delta(\sigma-\sigma')+2i e^{-2i\alpha}\,\delta_{ab}\,\delta'(\sigma-\sigma')\\
\{\bar J_a(\sigma),\bar J_b(\sigma')\}&=f_{ab}^c\, \bar J_c(\sigma)\,\delta(\sigma-\sigma')-2i e^{2i\alpha}\,\delta_{ab}\,\delta'(\sigma-\sigma')\\
\{J_a(\sigma),\bar J_b(\sigma')\}&=0.
\end{align*}
The Hamiltonian \eqref{ham} is
\begin{multline*}
\mathcal H=\half\int_{S^1}\biggl(\frac{\sin2\phi}{2}\la J(\sigma),\bar J(\sigma)\ra_\g\\ -\frac{\sin4\phi}{4}\Bigl(e^{2i\alpha}\la J(\sigma),J(\sigma)\ra_\g+e^{-2i\alpha}\la \bar J(\sigma),\bar J(\sigma)\ra_\g\Bigr)\biggr)d\sigma
\end{multline*}

A suitable Lagrangian Lie subalgebra $\h\subset\dd$ can be found as follows. Let $\mf n\subset\g_\C=\dd$ be the complex nilpotent Lie subalgebra spanned by the positive root spaces and let $\mf t\subset\g$ be the Cartan Lie subalgebra. Let $0\neq z\in\C$ be such that $\theta(z^2)=0$; up to a real multiple we have $z=e^{-i\alpha}$ or $z=ie^{-i\alpha}$. Then 
$$\h=z\mf t + \mf n\subset \g_\C=\dd$$
is a real Lie subalgebra of $\dd$ which is clearly Lagrangian. If $z\notin\R$ then $\h$ is transverse to $\g\subset\dd$ and we have an identification $D/H\cong G$ for the target space of the $\sigma$-model.

The case of $\alpha=0$ corresponds to Klimčík's Yang-Baxter $\sigma$-model \cite{kl}. The general case is the Yang-Baxter $\sigma$-model with WZW term introduced in \cite{dmv2} and reinterpreted as a $\sigma$-model of Poisson-Lie type in \cite{kl-ybwzw}.

%which gives
%$$a=\frac{\sin(\phi+\alpha)}{\sin2\phi}\qquad b=\frac{\sin(\phi-\alpha)}{\sin2\phi}.$$
%If $\phi=\pm\alpha$ mod $\pi\Z$ (i.e.\ if $ab=0$) we get, regardless of the choice of $\h\subset\dd=\g_\C$, the WZW model on $G$.
%
%If $j=j_\mathit{re}\otimes 1+j_\mathit{im}\otimes i$, let $J=j_\mathit{re}+ij_\mathit{im}$ and $\bar J=j_\mathit{re}-ij_\mathit{im}$. The Lax operator \eqref{lax} expressed in terms of $J$ and $\bar J$ is
%\begin{multline*}
%L(\sigma,\lambda)=\frac{2}{\sin^2 2\phi}\Bigl(
%\frac{\sin(\alpha-\phi)}{1+\lambda}\frac{e^{-i\alpha-i\phi}J(\sigma)-e^{i\alpha+i\phi}\bar J(\sigma)}{2i}\\
%+
%\frac{\sin(\alpha+\phi)}{1-\lambda}\frac{e^{-i\alpha+i\phi}J(\sigma)-e^{i\alpha-i\phi}\bar J(\sigma)}{2i}
%\Bigr)
%\end{multline*}
%
%$$\{J_a(\sigma),J_b(\sigma')\}=f_{ab}^c J_c(\sigma)\delta(\sigma-\sigma')+e^{\pm2i\alpha}\delta_{ab}\,\delta'(\sigma-\sigma')$$

The final case is $W=\R[\epsilon]/(\epsilon^2)$. After rescaling
 $\epsilon$ and $\la,\ra_\g$ we can suppose that $\theta(x+y\epsilon)=2tx+y$ for some $t\in\R$ and that
 $$e_+=1+(1-t)\epsilon\qquad e_-=1-(1+t)\epsilon.$$
Using the notation $j=j_1\otimes 1+j_\epsilon\otimes\epsilon$, we get
$$L(\sigma,\lambda)=\frac{1+t}{2(1+\lambda)}\bigl((1+t)j_1+j_\epsilon)\bigr)+\frac{1-t}{2(1-\lambda)}\bigl((1-t)j_1-j_\epsilon)\bigr).$$
The Poisson brackets are
\begin{align*}
\{j_{1a}(\sigma),j_{1b}(\sigma')\}&=f_{ab}^c\,j_{1c}(\sigma)\,\delta(\sigma-\sigma')\\
\{j_{1a}(\sigma),j_{2b}(\sigma')\}&=f_{ab}^c\,j_{2c}(\sigma)\,\delta(\sigma-\sigma')+\delta_{ab}\,\delta'(\sigma-\sigma')\\
\{j_{2a}(\sigma),j_{2b}(\sigma')\}&=-2t\,\delta_{ab}\,\delta'(\sigma-\sigma')
\end{align*}
and the Hamiltonian
$$\mathcal H=\half\int_{S^1}(1+t^2)\la j_0,j_0\ra_\g+t\la j_0,j_\epsilon\ra_\g+\la j_\epsilon,j_\epsilon\ra_\g$$

In this case the natural $\h\subset\dd=\g[\epsilon]/(\epsilon^2)$ is $\h=\epsilon\g$, which gives $D/H=G$. When $t=0$ the $\sigma$-model is the principal chiral model on $G$, when $t=\pm1$ we get the WZW model, and for other values of $t$ we get models given by the invariant metric on $G$ and by a multiple of the Cartan 3-form.

\end{document}